\documentclass[12pt,a4paper]{article}
\usepackage{latexsym,amssymb,amsmath,amsthm}
\usepackage[cp1251]{inputenc}

\newtheorem{theorem}{Theorem}
\newtheorem{lemma}{Lemma}

\newcommand{\rk}{{\rm rk_{\vee}\,}}
\newcommand{\rs}{{\rm rk_{+}\,}}
\newcommand{\rR}{{\rm rk_{\RR+}\,}}
\newcommand{\BB}{\mathbb{B}}
\newcommand{\NN}{\mathbb{N}}
\newcommand{\RR}{\mathbb{R}}
\newcommand{\OR}{\mathsf{OR}}
\newcommand{\SUM}{\mathsf{SUM}}
\newcommand{\ov}{\overline}
\newcommand{\Y}{{\cal Y}}

\begin{document}
\title{Upper bounds for the monotone rank of the unique disjointness matrix}
\date{}
\author{Igor S. Sergeev\footnote{e-mail: isserg@gmail.com}}

\maketitle

\begin{abstract}
It is shown that the $\OR$-rank (covering rank) of the $2^n \times 2^n$
unique disjointness matrix is $n^{O(1)}(3/2)^n$,
hence the known lower bound $1.5^n$ turns out to be essentially tight.
By the way, an upper bound $1.89^n$ is obtained for the $\SUM$-rank (partition rank) of this matrix.
\end{abstract}

\section{Introduction}

 The {\it unique disjointness matrix} $U_n$ (often denoted as $UDISJ_n$)
 is a partially defined $2^n \times 2^n$ boolean matrix.
 Its rows and columns are numbered by subsets of $[n]$, and the entries are defined as
\[ U_n[a,b] = \begin{cases} 1, & a \cap b = \varnothing \\ 0, & |a \cap b| = 1 \\ *, & |a \cap b| > 1 \end{cases}. \]
Further, when numbering the rows and columns of the matrix, in addition to subsets $a \subset [n]$,
we will also use their characteristic boolean vectors of length~$n$.

Although the unique disjointness matrix arises indirectly in Razborov's work~\cite{raz92e}
on a problem in the field of communication complexity, today results on the monotone rank
of this matrix find applications in the theory of optimization of linear programs; see, e.g.,~\cite{bfps15e,kw15e}.
However, many problems in these two areas (communication complexity and linear program optimization) are closely related,
are reduced to each other, or are studied using similar approaches; see, e.g.,~\cite{gjw18e}.

The {\it rank} of a matrix $A$ over a semiring $(S,+,\cdot)$ is defined as the minimum number $t$ of terms in a decomposition
\[ A = R_1 + \ldots + R_t, \qquad R_i =
   \begin{pmatrix} a_{i,1} \\ \vdots \\ a_{i,m} \end{pmatrix} \cdot \begin{pmatrix} b_{i,1} & \ldots & b_{i,n} \end{pmatrix},
   \quad  a_{i,j}, b_{i,k} \in S. \]
In the boolean case, $R_i$ are all-ones submatrices of $A$;
they are usually called (combinatorial) {\it rectangles}.
For partially defined boolean matrices, a rectangle is conveniently defined as a submatrix without zeros.

Let $\BB=\{0,\,1\}$. The main types of monotone rank of a (boolean partially defined) matrix $A$:

--- the rank $\rk A$ over the boolean semiring $(\BB, \vee, \wedge)$, also denoted as $\OR$-rank, the covering rank;

--- the rank $\rR A$ over the real semiring $(\RR_+, +, *)$, also known as the nonnegative rank;

--- the rank $\rs A$ over the integer semiring $(\NN_{\ge0}, +, *)$, also denoted as the $\SUM$-rank, the partition rank.

These quantities obey the (obvious) inequalities
\[ \rk A \le \rR A \le \rs A. \]

Covering and partition ranks characterize the communication complexity of a matrix:
the nondeterministic communication complexity of a matrix~$A$ is $\lceil \log_2 \rk A \rceil$,
and the unambiguous nondeterministic communication complexity is $\lceil \log_2 \rs A \rceil$.
In optimization problems with linear constraints, the nonnegative real rank is studied.
Due to the fundamental nature of the matrix, the question of its rank can be considered as being of independent interest.

Among the well-known completions of the matrix $U_n$ are the
Sierpi\'nski matrix\footnote{With entries $S_n[a,b] = (a \cap b = \varnothing)$.}~$S_n$ (the disjointness matrix)
and the complement of the boolean version of the Sylvester matrix\footnote{With entries $H_n[a,b] = |a \cap b| + 1 \bmod 2$.}~$H_n$.
These matrices have full $\OR$-rank, $\rk S_n = \rk H_n = 2^n$.

As usual, the {\it weight} of a (partially defined) boolean
matrix is defined as the number of ones in it.
An elegant argument from~\cite{kw15e} proves that any rectangle in~$U_n$
has weight~$\le 2^n$, so $\rk U_n \ge 1.5^n$, since the weight of the matrix is~$3^n$.
A simple construction from~\cite{we15e} shows that $\rk U_n = O(\sqrt3^{\:n})$.

In~\cite{bvg26e}, it was established that the method~\cite{we15e} yields an exact value
for the $\OR$-rank of~$U_n$ for $n \le 6$. However, below we will show that the lower bound~\cite{kw15e}
is in fact tight: $\rk U_n = n^{O(1)}(3/2)^n$.
Note that this result echoes a similar estimate for the extension complexity of the matching polytope
in~\cite{szu26e}\footnote{The solution to the latter problem involves estimating the nonnegative
real rank of a~somewhat similar matrix. An implicit analogy with the bounding
the rank of the unique disjointness matrix is explained in~\cite{sin17e}.}.
First, we will check that the $\SUM$-rank of~$U_n$ is also not full: ${\rs U_n \prec 1.89^{n}}$.

\section{Upper bound for the $\SUM$-rank}

Further, $||x||$ denotes the weight of a boolean vector $x$.
Recall that a {\it code} with distance $d$ (abbreviated, a $d$-code)
in $\BB^n$ is a set $X$ such that $|| x- y || \ge d$ for any distinct $x,y \in X$.

\begin{lemma}\label{ham}
The cube $\BB^n$ can be partitioned into $2n$ parts that are $3$-codes.
\end{lemma}

\proof Let $N$ be the nearest number of the form ${2^m-1}$ to $n$ from above.
We use the standard partition of the cube $\BB^N$ into balls of radius~1
centered at the points of the Hamming $3$-code $C_0$.
Denote by $C_i$ the set obtained from~$C_0$ by inverting $i$th coordinates
of the points~--- this is also a~$3$-code. By construction,
\[ \BB^N = C_0 \sqcup C_1 \sqcup \ldots \sqcup C_N. \]
Now it is easy to obtain the partition
\begin{equation}\label{code}
\BB^n = \Gamma_0 \sqcup \Gamma_1 \sqcup \ldots \sqcup \Gamma_N, \qquad \Gamma_i = \{ c \in \BB^n \mid (c,\vec 0) \in C_i \}.
\end{equation}
It remains to note that $N \le 2n-1$. \qed

Recall a simple relation for binomial coefficients:
\begin{equation}\label{bin}
\frac1{n} \cdot 2^{nH(k/n)} \preceq C_n^k \preceq 2^{nH(k/n)}, \qquad \text{as } \min\{k,\,n-k\} \to \infty,
\end{equation}
where $H(x) = -x\log_2x - (1-x)\log_2(1-x)$ is the {\it binary entropy function}
defined on the interval $[0,\,1]$; at the endpoints, $H(0)=H(1)=0$ by continuity.

\begin{theorem}\label{upsum}
$\rs U_n \preceq n^3\left(\frac3{\sqrt[3]4}\right)^{n} \prec
1.89^n.$
\end{theorem}

\proof Partition matrix $U_n$ into submatrices $U_n^{p,q}$ with
fixed row weights $p$ and columns weights $q$. We restrict ourselves to the case
${p+q \le n}$ (otherwise, the submatrix contains no ones). For each matrix $U_n^{p,q}$,
we construct a partition into rectangles in one of two ways.

In the first way, apply a trivial partition by rows or by columns, depending on
whether $p$ or $q$ is farther from $n/2$. The cardinality of such a partition does not exceed
\[ \min\{ C_n^p,\, C_n^q \} \le C_n^{(p+q)/2}. \]

In the second way, a matrix $U_n^{p,q}$ is first partitioned into submatrices~$U^S$,
whose rows and columns belong to (all possible) sets~$S$ of cardinality~$p+q$.
To construct a partition of a matrix $U^S$, apply partition~(\ref{code})
of the boolean cube $\BB^{p+q}$ from Lemma~\ref{ham}.
From any code set~$\Gamma_i$, select the subset $R_i$ of all vectors of weight $p$.
Then $R_i \times \ov{R_i}$ defines the desired partition rectangle,
where $\ov X$ denotes the elementwise inverse of the set $X$.
The fact that the rectangle is legal is ensured by the property of $R_i$ to be a 3-code.
Indeed, since $|| x - y || \ge 3$ for distinct $x,y \in R_i$,
vectors $x$ and $\ov y$ have at least two common ones.
The cardinality of the proposed partition of the matrix $U_n^{p,q}$ does not exceed $2n C_n^{p+q}$.

By (\ref{bin}), we obtain
\[ \rs U_n^{p,q} \le \min\{ C_n^{(p+q)/2},\, 2n C_n^{p+q} \} \preceq n C_n^{n/3} \le n2^{nH(1/3)} = n\left(\frac3{\sqrt[3]4}\right)^{n}, \]
since $H(1/3) = \log_23-2/3$.

It remains to sum the ranks of all $O(n^2)$ submatrices $U_n^{p,q}$. \qed

\section{Upper bound for the $\OR$-Rank}

A technical ingredient for the proof is the Sapozhenko---Lovasz---Stein lemma
on the rank of a gradient covering~\cite{sap72e,ste74e}, see also~\cite[\S2.6]{juk11e}.
In a convenient and slightly weakened form, it is stated as follows.

\begin{lemma}\label{grad}
Let ${\cal F} \subset 2^X$.
If the elements of $X$ are uniformly covered by sets from the family $\cal F$,
then $\cal F$ contains a covering of $X$ of size $w \le {(1+\ln |X|)|X|/s}$,
where $s$ is the average cardinality of a set in $\cal F$.
\end{lemma}

Further, $X^{(k)}$ denotes the family of all subsets of cardinality~$k$ in $X$.

\begin{theorem}
$\rk U_n \le n^{O(1)}(3/2)^{n}.$
\end{theorem}

\proof Consider again the partition of matrix $U_n$ into submatrices $U_n^{p,q}$
with row weights $p$ and column weights $q$. We will construct a covering for
each matrix $U_n^{p,q}$ independently. Without loss of generality (the argument is symmetric),
assume $p \ge q$. Denote $\alpha = p/n$ and $\beta = q/n$.

{\bf I.} Degenerate case: $\beta \le 0.1$ or $\alpha \ge 0.9$.

Even the trivial partition of a matrix $U_n^{p,q}$ by rows or by columns has rank
\[ C_n^{\min \{q,\, n-p \}} \preceq 2^{nH(0.1)} \prec 1.4^n.\]

In what follows, we assume that $\alpha, \beta \in [0.1, 0.9]$.

{\bf II.} Main case: $\alpha \le 1/2$ and $\alpha+\beta \ge 1/2$.

We define a family of rectangles in $U_n^{p,q}$. It is obtained by combining
the method of Lemma~\ref{ham} and the standard technique of separation
of variables\footnote{The latter approach allows to construct optimal coverings of matrices $S_n$, see~\cite{cils17e}.}.
Set $x = \alpha+\beta - 1/2$. Randomly partition $[n]$ into three groups:
$X_0$ of $2xn$ elements, $X_1$ of $2(p-xn)$ elements, and $X_2$ of
$2(q-xn)$ elements\footnote{From here on, we neglect parity and roundings.
Taking these factors into account only results in a refinement of the $n^{O(1)}$ factor in the estimates.}.
Applying~(\ref{code}), we obtain a partition of the family $X_0^{(xn)}$ into 3-codes:
\[ X_0^{(xn)} = \Y_1 \sqcup \Y_2 \sqcup \ldots \sqcup \Y_{4xn}. \]

The row indices of a rectangle $R_{X_1,X_2,i}$ are composed of all possible combinations 
$a_0 \cup a_1$, where $a_0 \in \Y_i$ and $a_1 \in X_1^{(p-xn)}$, and the column indices
are composed of combinations $b_0 \cup b_2$, where $b_0 \in X_0 \setminus \Y_i$ 
(the set difference operation is applied to each set in $\Y_i$) and $b_2 \in X_2^{(q-xn)}$.

The family ${\cal F} = \{ R_{X_1,X_2,i} \}$ covers the ones of $U_n^{p,q}$ 
uniformly\footnote{The family is invariant under permutations in $[n]$.}.
The average weight of a rectangle is
\[ n^{\pm O(1)} C_{2xn}^{\,xn} C_{2(p-xn)}^{\,p-xn} C_{2(q-xn)}^{\,q-xn} = n^{- O(1)} 2^n. \]
Therefore, according to Lemma~\ref{grad}, $\cal F$ contains a covering of size
\[ n^{O(1)} |U_{n}^{p,q}|/2^n \le n^{O(1)}\left(3/2\right)^{n}. \]

{\bf III.} The case $\alpha+\beta \le 1/2$.

To construct a covering family, we use only separation of variables. Denote $\gamma = \frac1{\alpha+\beta}$. 
Randomly partition $[n]$ into two parts: $X_1$ of $\gamma p$ elements and $X_2$ of $\gamma q$ elements. 
The rows and columns of a rectangle $R_{X_1}$ are numbered 
by sets from $X_1^{(p)}$ and $X_2^{(q)}$, respectively. The weight of any rectangle is
\[ n^{\pm O(1)} C_{\gamma p}^{p} C_{\gamma q}^{q} = n^{\pm O(1)} 2^{n H(\alpha+\beta)}. \]
The family ${\cal F} = \{ R_{X_1} \}$ covers the ones uniformly, hence it contains a covering of size
\[
n^{\pm O(1)} |U_{n}^{p,q}|/2^{n H(\alpha+\beta)} = n^{\pm O(1)}
C_n^{p+q} C_{p+q}^p /2^{n H(\alpha+\beta)} = n^{\pm O(1)}
C_{p+q}^p \le n^{O(1)} \sqrt2^{\,n}.
\]

{\bf IV.} The case $\alpha \ge 1/2$.

To construct the family, we employ a combination of the two methods of Theorem~\ref{upsum}. 
Let $z = \frac{\alpha\beta}{1-\alpha}$.
Randomly choose a subset~$X$ of $q+zn$ elements in~$[n]$.
Consider a partition of the family $X^{(q)}$ into 3-codes:   % (Lemma~\ref{ham})
\[ X^{(q)} = \Y^X_1 \sqcup \Y^X_2 \sqcup \ldots \sqcup \Y^X_{2(q+zn)}. \]
A rectangle $R_{X,i}$ is formed by the columns $\Y^X_i$, and its row indices are 
all possible sets $a_0 \cup a_1$ of cardinality $p$, where $a_0 \in X \setminus \Y^X_i$ 
(the set difference operation is applied to every set in $\Y^X_i$) and $a_1 \subset [n] \setminus X$.

The proposed family ${\cal F} = \{ R_{X,i} \}$ uniformly covers the ones of~$U_n^{p,q}$. 
The average weight of a rectangle is
\[ n^{\pm O(1)} C_{q+zn}^{\,zn} C_{n-q-zn}^{\,p-zn} = n^{\pm O(1)} 2^{nH(\alpha)}, \]
since $\frac{zn}{q+zn} = \frac{p-zn}{n-q-zn} = \alpha$ due to the choice of $z$. 
Lemma~\ref{grad} guarantees in ${\cal F} = \{ R_{X,i} \}$ a covering for $U_n^{p,q}$ of size
\[ n^{\pm O(1)} |U_{n}^{p,q}|/2^{nH(\alpha)} = n^{\pm O(1)} C_n^{\,p} C_{n-p}^{\,q} /2^{nH(\alpha)} = n^{\pm O(1)} C_{n-p}^{\,q} \le n^{O(1)} \sqrt2^{\,n}. \]

Summing up the ranks of all submatrices yields the final estimate. \qed


\begin{thebibliography}{99}

\bibitem{bvg26e}
Baeckelant T., Vandaele A., Gillis N. {\it Computing lower bounds on the nonnegative rank via
non-convex optimization solvers.} 2026. {\sf arXiv:2605.14058}.

\bibitem{bfps15e}
Braun G., Fiorini S., Pokutta S., Steurer D. {\it Approximation limits of linear programs $($beyond hierarchies$)$.}
Math. Operations Research. 2015. {\bf40}(3), 756--772.

\bibitem{cils17e}
Chistikov D., Iv\'an S., Lubiw A., Shallit J. {\it Fractional
coverings, greedy coverings, and rectifier networks.} Proc. STACS % Symp. on Theor. Aspects of Comput. Sci.
(Hannover, 2017). LIPIcs. 2017. {\bf66}, Art.~23.

\bibitem{gjw18e}
G\"o\"os M., Jain R., Watson T. {\it Extension complexity of independent set polytopes}.
SIAM J. Comput. 2018. {\bf47}(1), 241--269.

\bibitem{juk11e}
Jukna S. {\it Extremal combinatorics: with applications in
computer science.} Berlin, Heidelberg: Springer--Verlag, 2011.

\bibitem{kw15e}
Kaibel V., Weltge S. {\it A short proof that the extension complexity
of the correlation polytope grows exponentially.} Discrete Comput.
Geometry. 2015. {\bf53}(2), 397--401.

\bibitem{raz92e}
Razborov A. A. {\it On the distributional complexity of disjointness.}
Theoret. Comput. Sci. 1992. {\bf106}(2), 385--390.

\bibitem{sap72e}
Sapozhenko A. A. {\it On the complexity of disjunctive normal
forms obtained with a gradient algorithm.} In: Diskretnyj
Analiz [Discrete Analysis]. Vol.~21. Novosibirsk: Inst.
Matem. SO AN SSSR, 1972, 62--71. (in Russian)

\bibitem{sin17e}
Sinha M. {\it Lower bounds for approximating the matching polytope.}
Proc. SODA (New Orleans, 2018). SIAM, 2018, 1585--1604.
%2017. arXiv:1711.10145.

\bibitem{ste74e}
Stein S. K. {\it Two combinatorial covering problems.} J. Combin.
Theory~(A). 1974. {\bf16}, 391--397.

\bibitem{szu26e}
Szusterman M. {\it Markovian protocols and an upper bound on the extension
complexity of the matching polytope.} 2026. {\sf arXiv:2602.11382}.

\bibitem{we15e}
Weltge S. {\it Sizes of linear descriptions in combinatorial optimization.} PhD thesis,
Otto-von-Guericke-Universit\"at Magdeburg, 2015.

\end{thebibliography}
\end{document}

\title{Верхние оценки монотонного ранга матрицы 1-непересечений}
\date{}
\author{И. С. Сергеев\footnote{e-mail: isserg@gmail.com}}

\maketitle

\begin{abstract}
Показано, что $\OR$-ранг (ранг минимального покрытия) $2^n \times
2^n$ матрицы 1-непересечений равен $n^{O(1)}(3/2)^n$: таким образом,
известная нижняя оценка $1.5^n$ оказывается по сути точной.
Попутно для $\SUM$-ранга (ранга минимального разбиения) этой
матрицы получена верхняя оценка вида $1.89^n$.
\end{abstract}

\section{Введение}

{\it Матрица $1$-непересечений} (unique disjointness matrix) $U_n$ (часто
обозначается как $UDISJ_n$) --- это частично определенная $2^n \times 2^n$
булева матрица. Ее строки и столбцы нумеруются подмножествами множества $[n]$, а коэффициенты
определяются как
\[  U_n[a,b] = \begin{cases} 1, & a \cap b = \varnothing \\ 0, & |a \cap b| = 1 \\ *, & |a \cap b| > 1 \end{cases}.  \]
Далее при нумерации строк и столбцов матрицы мы также будем помимо
множеств $a \subset [n]$ использовать их характеристические булевы
вектора длины $n$.

Хотя матрица 1-непересечений возникает косвенным образом в работе Разборова~\cite{raz92},
посвященной проблеме из области коммуникационной сложности, сегодня результаты о монотонном ранге
матрицы 1-непересечений находят приложения в теории оптимизации линейных программ, см., например,~\cite{bfps15,kw15}.
Впрочем, многие задачи из указанных областей
(коммуникационная сложность и оптимизация линейных программ) близки,
сводятся к друг к другу или исследуются схожими методами,
см., например,~\cite{gjw18}.

% теории сложности расширений политопов, которая послужили основной мотивацией для работы~\cite{kw15}).

{\it Ранг} $m \times n$ матрицы $A$ над полукольцом $(S,+,\cdot)$ определяется как минимальное число $t$ слагаемых
в декомпозиции
\[  A = R_1 + \ldots + R_t, \qquad  R_i =
    \begin{pmatrix} a_{i,1}  \\ \vdots \\ a_{i,m} \end{pmatrix} \cdot \begin{pmatrix} b_{i,1} & \ldots & b_{i,n} \end{pmatrix},
   \quad  a_{i,j}, b_{i,k} \in S.  \]
В булевом случае $R_i$ --- это сплошь единичные подматрицы матрицы $A$, они обычно называются
(комбинаторными) {\it прямоугольниками}. Для частично определенных булевых матриц прямоугольником
удобно считать подматрицу без нулей.

Пусть $\BB=\{0,\,1\}$.  Основные виды монотонного ранга (булевой частично определенной) матрицы $A$:

--- ранг $\rk A$ над булевым полукольцом $(\BB, \vee, \wedge)$, он же $\OR$-ранг,
ранг покрытия (прямоугольниками);

--- ранг $\rR A$ матрицы над вещественным полукольцом $(\RR_+, +, *)$, он же неотрицательный ранг;

--- ранг $\rs A$ над целочисленным полукольцом $(\NN_{\ge0}, +, *)$, он же $\SUM$-ранг,
ранг разбиения (на прямоугольники).

Указанные величины подчиняются (очевидным) соотношениям
\[ \rk A  \le \rR A \le \rs A.  \]

Ранги покрытия и разбиения характеризуют коммуникационную сложность матрицы:
недетерминированная коммуникационная
сложность матрицы $A$ равна $\lceil \log_2 \rk A \rceil$, а
однозначная недетерминированная коммуникационная сложность равна
$\lceil \log_2 \rs A \rceil$. В задачах оптимизации с линейными ограничениями изучается
неотрицательный вещественный ранг. Ввиду фундаментальной природы матрицы,
вопрос о ее ранге представляет вполне самостоятельный интерес.

Среди известных доопределений матрицы $U_n$ --- матрица
Серпинского\footnote{Со значениями $S_n[a,b] = (a \cap b =
\varnothing)$.}~$S_n$ (матрица непересечений --- disjointness matrix) и дополнение булевой
версии матрицы Сильвестра\footnote{Со значениями $H_n[a,b] = |a
\cap b| + 1 \bmod 2$.}~$H_n$. Обе указанных матрицы имеют полный
$\OR$-ранг, $\rk S_n = \rk H_n = 2^n$.

Как обычно, под {\it весом} (частично определенной) булевой
матрицы понимается число единиц в ней.
Изящное рассуждение из~\cite{kw15} доказывает, что любой
прямоугольник в~$U_n$ имеет вес~$\le 2^n$, поэтому $\rk
U_n \ge 1.5^n$, так как вес матрицы равен~$3^n$. Простая
конструкция из~\cite{we15} показывает, что $\rk U_n =
O(\sqrt3^{\:n})$.

В~\cite{bvg26} установлено, что метод~\cite{we15} дает точное
значение $\OR$-ранга~$U_n$ при $n \le 6$. Тем не менее далее мы покажем,
что на самом деле точна нижняя оценка~\cite{kw15}: справедливо $\rk U_n =
n^{O(1)}(3/2)^n$. Отметим, что полученный результат перекликается
с аналогичной оценкой сложности расширения многогранника паросочетаний
из~\cite{szu26}\footnote{Решение последней задачи связано с оценкой монотонного
вещественного ранга в какой-то степени похожей матрицы. Неявная аналогия с оценкой ранга
матрицы 1-непересечений объясняется в~\cite{sin17}.}.
Предварительно мы проверим, что и $\SUM$-ранг матрицы~$U_n$
не полон: $\rs U_n \prec 1.89^{n}$.

\section{Верхняя оценка $\SUM$-ранга}

Далее через $||x||$ обозначаем вес булева вектора $x$. Напомним, что {\it кодом} с расстоянием
$d$ (сокращенно, $d$-кодом) в~$\BB^n$ называется множество $X$ с~условием $|| x- y || \ge d$ для любых различных $x,y \in X$.

\begin{lemma}\label{ham}
Куб $\BB^n$ можно разбить на $2n$ штук $3$-кодов.
\end{lemma}

\proof Пусть $N$ --- ближайшее к $n$ сверху число вида ${2^m-1}$.
Воспользуемся стандартным разбиением куба $\BB^N$ на шары радиуса
1 с центрами в точках $3$-кода Хэмминга $C_0$. Через $C_i$
обозначим множество, получаемое из $C_0$ инвертированием $i$-х
координат точек --- это тоже $3$-код. По построению,
\[  \BB^N = C_0 \sqcup C_1  \sqcup \ldots \sqcup C_N. \]
Отсюда легко перейти к разбиению
\begin{equation}\label{code}
 \BB^n = \Gamma_0 \sqcup \Gamma_1 \sqcup \ldots \sqcup \Gamma_N, \qquad \Gamma_i = \{ c \in \BB^n \mid (c,\vec 0) \in C_i \}.
\end{equation}
Остается заметить, что $N \le 2n-1$. \qed

Напомним простое соотношение для биномиальных коэффициентов:
\begin{equation}\label{bin}
 \frac1{n} \cdot 2^{nH(k/n)} \preceq C_n^k \preceq 2^{nH(k/n)}, \qquad \text{при } \min\{k,\,n-k\} \to \infty,
\end{equation}
где $H(x) = -x\log_2x - (1-x)\log_2(1-x)$ --- {\it функция
двоичной энтропии}, определенная на отрезке $[0,\,1]$; на
концах отрезка полагается по непрерывности $H(0)=H(1)=0$.

\begin{theorem}\label{upsum}
$\rs U_n \preceq n^3\left(\frac3{\sqrt[3]4}\right)^{n} \prec
1.89^n.$
\end{theorem}

\proof Разобьем матрицу $U_n$ на подматрицы $U_n^{p,q}$ с
фиксированным весом строк $p$ и столбцов $q$. Рассмотрение
ограничивается случаем $p+q \le n$ (иначе подматрица не содержит
единиц). Для каждой матрицы $U_n^{p,q}$ построим разбиение на
прямоугольники одним из двух способов.

В первом способе применяем тривиальное разбиение по строкам или по столбцам в зависимости от того,
$p$ или $q$ находится дальше от $n/2$. Мощность такого разбиения не превосходит
\[ \min\{ C_n^p,\, C_n^q \} \le C_n^{(p+q)/2}. \]

Во втором способе матрица $U_n^{p,q}$ сначала разбивается на подматрицы~$U^S$,
строки и столбцы которых принадлежат всевозможным множествам~$S$ мощности~$p+q$.
Для построения разбиения матрицы $U^S$ применяем
разбиение (\ref{code}) булева куба $\BB^{p+q}$ из леммы~\ref{ham}.
Из кодового множества~$\Gamma_i$ выделим подмножество $R_i$ всех
векторов веса $p$. Тогда $R_i \times \ov{R_i}$ задает искомый %(S \setminus R_i)
прямоугольник разбиения, где $\ov X$ означает поэлементную инверсию множества $X$ .
То, что прямоугольник является легальным,
обеспечивает свойство $R_i$ быть 3-кодом. Действительно, поскольку
$|| x - y || \ge 3$ для различных $x,y \in R_i$, вектора $x$ и
$\ov y$ имеют более одной общей единицы. Общая мощность разбиения матрицы $U_n^{p,q}$
не превосходит $2n C_n^{p+q}$.

В силу (\ref{bin}) получаем
\[  \rs U_n^{p,q} \le \min\{ C_n^{(p+q)/2},\, 2n C_n^{p+q} \} \preceq n C_n^{n/3} \le n2^{nH(1/3)} = n\left(\frac3{\sqrt[3]4}\right)^{n}, \]
поскольку $H(1/3) = \log_23-2/3$.

Осталось суммировать ранги всех $O(n^2)$ подматриц $U_n^{p,q}$. \qed

\section{Верхняя оценка $\OR$-ранга}

Техническим инструментом доказательства служит лемма
Сапоженко---Ловаса---Стейна о ранге градиентного
покрытия~\cite{sap72,ste74}, см. также~\cite[\S2.6]{juk11}. В~удобном для
применения (слегка ослабленном) виде она формулируется так.

\begin{lemma}\label{grad}
Пусть ${\cal F} \subset 2^X$.
%Если %$|F| \le s$ для любого $F \in \cal F$, и
%любой элемент $x \in X$ принадлежит не менее чем $t$ множествам из
%$\cal F$, то семейство $\cal F$ содержит покрытие множества $X$
%размера $w \le (1+\ln |X|)|{\cal F}|/t$.
Если элементы $X$ равномерно накрываются множествами из $\cal F$,
то $\cal F$ содержит покрытие множества $X$ размера $w \le (1+\ln
|X|)|X|/s$, где $s$ --- средняя мощность множества из $\cal F$.
\end{lemma}

Далее через $X^{(k)}$ обозначается семейство всех подмножеств
мощности~$k$ в множестве $X$.

\begin{theorem}
$\rk U_n \le n^{O(1)}(3/2)^{n}.$
\end{theorem}

\proof Опять рассмотрим разбиение матрицы $U_n$ на подматрицы
$U_n^{p,q}$ с весом строк $p$ и столбцов $q$. Покрытие каждой
матрицы $U_n^{p,q}$ будем строить независимо. Без ограничения
общности (рассуждение симметрично) полагаем $p \ge q$. Обозначим
$\alpha = p/n$ и $\beta = q/n$.

{\bf I.} Вырожденный случай: $\beta \le 0.1$ или $\alpha \ge 0.9$.

Уже тривиальное разбиение матрицы $U_n^{p,q}$ по строкам или по столбцам имеет
ранг
\[ C_n^{\min \{q,\, n-p \}} \preceq 2^{nH(0.1)} \prec 1.4^n.\]

Далее, особо не оговаривая, полагаем, что $\alpha, \beta \in [0.1, 0.9]$.

{\bf II.} Основной случай: $\alpha \le 1/2$ и $\alpha+\beta \ge 1/2$.

Определим семейство прямоугольников в $U_n^{p,q}$. Они строятся
комбинированием способа леммы~\ref{ham} и обычного приема разделения
переменных\footnote{Такой подход применяется при построении
оптимальных покрытий матриц $S_n$, см.~\cite{cils17}.}. Положим $x
= \alpha+\beta - 1/2$. Случайно разобьем $[n]$ на 3 группы: $X_0$
из $2xn$ элементов, $X_1$ из $2(p-xn)$ элементов и $X_2$ из
$2(q-xn)$ элементов\footnote{Здесь и далее пренебрегаем четностью
и округлениями. Учет этих факторов выражается только в уточнении
показателя в множителе $n^{O(1)}$ в оценках.}. При помощи
конструкции (\ref{code}) получим разбиение семейства~$X_0^{(xn)}$
на 3-коды:
\[ X_0^{(xn)} = \Y_1 \sqcup \Y_2  \sqcup \ldots \sqcup \Y_{4xn}.  \]

Номера строк прямоугольника $R_{X_1,X_2,i}$ составим из
всевозможных комбинаций $a_0 \cup a_1$, где $a_0 \in \Y_i$ и $a_1
\in X_1^{(p-xn)}$, а номера столбцов --- из комбинаций $b_0 \cup
b_2$, где $b_0 \in X_0 \setminus \Y_i$ (операция разности
применяется к каждому множеству в $\Y_i$) и $b_2 \in X_2^{(q-xn)}$.

Семейство ${\cal F} = \{ R_{X_1,X_2,i} \}$ покрывает единицы
$U_n^{p,q}$ равномерно\footnote{Семейство инвариантно относительно
перестановок в $[n]$.}. Средний вес прямоугольника имеет величину
\[  n^{\pm O(1)} C_{2xn}^{\,xn}  C_{2(p-xn)}^{\,p-xn} C_{2(q-xn)}^{\,q-xn} = n^{- O(1)} 2^n. \]
Поэтому согласно лемме~\ref{grad} в $\cal F$ можно выделить
покрытие мощности
\[  n^{O(1)} |U_{n}^{p,q}|/2^n \le  n^{O(1)}\left(3/2\right)^{n}. \]

{\bf III.} Случай $\alpha+\beta \le 1/2$.

Для построения покрывающего семейства используем только разделение
переменных. Обозначим $\gamma = \frac1{\alpha+\beta}$. Множество
$[n]$ случайно разбивается на 2 части: $X_1$ из %$\frac{p}{\alpha+\beta}$
$\gamma p$ элементов и $X_2$ из $\gamma q$ элементов. Строки и
столбцы прямоугольника $R_{X_1}$ нумеруются множествами из
$X_1^{(p)}$ и $X_2^{(q)}$ соответственно. Вес любого
прямоугольника равен
\[  n^{\pm O(1)} C_{\gamma p}^{p} C_{\gamma q}^{q} = n^{\pm O(1)} 2^{n H(\alpha+\beta)}. \]
Семейство ${\cal F} = \{ R_{X_1} \}$ покрывает единицы равномерно,
поэтому в нем можно выделить покрытие размера
\[
 n^{\pm O(1)} |U_{n}^{p,q}|/2^{n H(\alpha+\beta)} = n^{\pm O(1)}
 C_n^{p+q} C_{p+q}^p /2^{n H(\alpha+\beta)} = n^{\pm O(1)}
 C_{p+q}^p \le n^{O(1)} \sqrt2^{\,n}.
\]

{\bf IV.} Случай $\alpha \ge 1/2$.

Для построения семейства используем комбинацию двух способов
теоремы~\ref{upsum}. Обозначим $z = \frac{\alpha\beta}{1-\alpha}$.
Случайно выберем в $[n]$ подмножество~$X$ из $q+zn$ элементов.
Зададим разбиение семейства $X^{(q)}$ на 3-коды:   %(лемма~\ref{ham})
\[ X^{(q)} = \Y^X_1 \sqcup \Y^X_2  \sqcup \ldots \sqcup \Y^X_{2(q+zn)}.  \]
Прямоугольник $R_{X,i}$ образован столбцами $\Y^X_i$, а номера
строк заданы всевозможными множествами $a_0 \cup a_1$ мощности
$p$, где $a_0 \in X \setminus \Y^X_i$  %\footnote{Как и выше,.}
(операция разности применяется к
каждому множеству из семейства $\Y^X_i$) и $a_1 \subset [n] \setminus X$.

Построенное семейство ${\cal F} = \{ R_{X,i} \}$ равномерно
покрывает единицы~$U_n^{p,q}$. Средний вес прямоугольника имеет
величину
\[  n^{\pm O(1)} C_{q+zn}^{\,zn} C_{n-q-zn}^{\,p-zn} = n^{\pm O(1)} 2^{nH(\alpha)}, \]
поскольку $\frac{zn}{q+zn} = \frac{p-zn}{n-q-zn} = \alpha$
согласно выбору $z$. Лемма~\ref{grad} гарантирует в ${\cal F} = \{
R_{X,i} \}$ наличие покрытия для $U_n^{p,q}$ размером
\[  n^{\pm O(1)} |U_{n}^{p,q}|/2^{nH(\alpha)} = n^{\pm O(1)} C_n^{\,p} C_{n-p}^{\,q} /2^{nH(\alpha)} =  n^{\pm O(1)}  C_{n-p}^{\,q} \le  n^{O(1)} \sqrt2^{\,n}. \]

Сложение рангов всех подматриц дает окончательную оценку. \qed